\documentclass{eptcs}
\providecommand{\event}{QAPL 2013} % Name of the event you are submitting to
\usepackage{breakurl}             % Not needed if you use pdflatex only.

\usepackage{cite}

\title{Probabilistic Concurrent Kleene Algebra}
\author{
Annabelle McIver \qquad\qquad Tahiry Rabehaja
\institute{Department of Computing~\thanks{This research has been supported by the  Australia Research Council Discovery Grant DP1092464 and the iMQRS Grant from Macquarie University.}\\
Macquarie University \\
Sydney, Australia}
\email{\{annabelle.mciver,tahiry.rabehaja\}@mq.edu.au}
\and
Georg Struth
\institute{Department of computer Science\\
University of Sheffield\\
United Kingdom}
\email{g.struth@dcs.shef.ac.uk}
}
\def\titlerunning{Probabilistic Concurrent Kleene Algebra}
\def\authorrunning{Annabelle McIver, Tahiry Rabehaja, Georg Struth}

\usepackage{xy}
\xyoption{all}

\usepackage{amsthm}
\newtheorem{theorem}{Theorem}[section]
\newtheorem{definition}[theorem]{Definition}
\newtheorem{proposition}[theorem]{Proposition}
\newtheorem{corollary}[theorem]{Corollary}
\newtheorem{example}[theorem]{Example}

\usepackage{amssymb}

\newcommand{\supp}{\mathrm{supp}}

\newcommand{\tea}{\mathtt{tea}}
\newcommand{\coin}{\mathtt{coin}}
\newcommand{\kick}{\mathtt{kick}}
\newcommand{\fail}{\mathtt{fail}}
\newcommand{\stuck}{\mathtt{stuck}}

\newcommand{\last}{\mathrm{last}}

\newcommand{\run}{\mathbf{run}}

\newcommand{\DD}{\mathcal{D}}
\newcommand{\D}{\Delta}

\newcommand{\F}{\mathcal{F}}

\newcommand{\lra}{\longrightarrow}
\newcommand{\Lra}{\Longrightarrow}

\newcommand{\arp}{\ar@{.}}

\newcommand{\ov}[1]{\overline{#1}}
\newcommand{\ovov}[1]{\overline{\overline{#1}}}
\newcommand{\pc}[1]{{\ \empty_{#1}}\!\!\oplus\ \!}
\newcommand{\pr}[1]{{\empty_{#1}}\!\|}

\newcommand{\trans}[3]{#1\stackrel{#2}{\lra}#3}
\newcommand{\Trans}[3]{#1\stackrel{#2}{\Lra}#3}

\newcommand{\unfold}{\mathbf{unfold}}
\newcommand{\Path}{\mathrm{Path}}
\newcommand{\Deadlock}{\mathbf{deadlock}}
\newcommand{\zero}{\mathbf{0}}
\newcommand{\Skip}{\mathbf{skip}}
\newcommand{\one}{\mathbf{1}}
\newcommand{\Pa}{\mathbf{a}}

\newcommand{\pAut}{\mathbf{PAut}}

\begin{document}
\maketitle

\begin{abstract}
We provide an extension of concurrent Kleene algebras to account for probabilistic properties. The algebra yields a unified framework containing nondeterminism, concurrency and probability and is sound with respect to the set of probabilistic automata modulo probabilistic simulation. We use the resulting algebra to generalise the algebraic formulation of a variant of Jones' rely/guarantee calculus.
\end{abstract}

\section{Introduction}

Since Rabin's seminal paper introducing probabilistic algorithms~\cite{Rab76}, the role of probability in the design of systems has increased in popularity and effectiveness. However, probability and concurrency results in greater complexity and so the formal verification of such system poses a particular challenge.  Many process algebras and automated tools have been developed to combine probability and concurrency but there are some important features of concurrent systems such as interference between shared variables that are not easily captured. Interference is for example a feature that appears in Jones' rely/guarantee calculus. In this paper, we generalise concurrent Kleene algebra and extend the variant of Jones' rely/guarantee rules found in~\cite{Hoa09} to apply to systems with probability. 

Introduced by Kleene in his study of automata and regular languages~\cite{Kle51}, Kleene algebras provide an elegant tool that is able to express static properties of standard sequential programs~\cite{Koz97,Koz00a,Koz03}. The axiomatisation of Kleene algebras and their extensions have been well studied by Conway~\cite{Con71}, Kozen~\cite{Koz94},  Salomaa~\cite{Sal66} and others~\cite{Des06,Mo07}.

Concurrent Kleene algebra has been developed and used to give robust proofs of concurrent systems, and in particular for verification techniques such as a variant of Jones' rely/guarantee rules~\cite{Jon81,Hoa09}. Algebra is an important mathematical tool for carrying out complex proofs, however its effectiveness relies on realistic models. In this paper, we give the first Kleene like structure that extends concurrent Kleene algebra with probabilistic behaviour. 

Our algebra is constructed from concurrent and probabilistic Kleene algebras~\cite{Mci04,Hoa09}. It provides an abstraction of concurrent systems and is able to express properties such as interference. The most important rule of concurrent Kleene algebra is the interchange law which establishes the interaction between concurrency and sequential execution. Informally, it says that the sequential execution of two concurrent systems entails an explicit dependency between the two concurrent parts, therefore no interference. But the concurrent execution of two sequential systems allows interference between any parts of the two systems. The relationship between these behaviours is expressed within a partial order determined by traces, and the algebraic rules mandate program operators' behaviours when used in combination.

In contrast, probabilistic Kleene algebra provides an abstraction of probability whose existence can be deduced from a subdistributivity law. The law describes the interaction between sequential execution and nondeterminism in the presence of probability. Informally, it says that a nondeterministic choice preceded by a probabilistic action can be resolved using the outcome of the probabilistic choice, that is,
$$x\cdot y + x\cdot z\leq x\cdot(y+z)$$
where $x,y,z$ are probabilistic programs, $+$ is the nondeterministic choice and $\cdot$ denotes sequential execution. This law also ensure monotonicity of the sequential composition which is an important property for combining sequential system from smaller components.

Our first contribution is to expand the set of axioms for concurrent Kleene algebra to account for the presence of probability and prove their soundness with respect to an automata model of probabilistic concurrency modulo simulation equivalence. In proving the soundness we find that some of the original axioms, both from concurrent and probabilistic Kleene algebras, needed to be weakened to accommodate the presence of both features in a single framework. Our simulation is based on the definition of Deng et al.~\cite{Den07a} though we show that it is equivalent to Segala\rq{}s definition of probabilistic weak forward simulation~\cite{Seg95}. Segala\rq{}s simulation completely characterises the coarsest precongruence included in the trace distribution equivalence of probabilistic automata~\cite{Seg03} and coincides with the infinitary probabilistic vector may testing order~\cite{Seg96}.  

Our second contribution is the extension of the rely/guarantee calculus of~\cite{Hoa09} to concurrent systems exhibiting probabilistic behaviours. 
We show that with the axiomatisation presented in Section~\ref{sec:soundness}, the algebra is sufficient to prove some important rely/guarantee rules which hold in the action-based interleaving model as well as any other semantics that satisfy the necessary algebraic properties. Since the inequality in the rules can be interpreted as the existence of a probabilistic simulation, the testing interpretation of simulation~\cite{Den07b} allows us to provide bounds for the maximal probability of failure.

In Section~\ref{sec:p-automata}, we provide a summary of probabilistic automata and the operators which will be used to interpret our algebraic terms. Section~\ref{sec:simulation} contains a survey of probabilistic simulation where the formulation of Deng et al is shown to be equivalent to Segala\rq{}s definition of probabilistic weak forward simulation. Section~\ref{sec:soundness} provides the soundness of the algebra with respect to automata and probabilistic simulation. In particular, we show how the probability interacts with the sequential as well as concurrent operators. We also provide the derivation of some rely/guarantee rules within and study a simple example to show the application of the algebraic reasoning developed. 

\section{Probabilistic Automata and Operations}\label{sec:p-automata}
The standard constructions of automata theory have been generalised to capture probabilistic behaviour which we summarise briefly here. We will use these automata as a model to exhibit soundness for the generalised concurrent
Kleena Algebra we describe in Section 4.

Probability is encoded in terms of distributions over the state space. A transition in a probabilistic automaton starts from a source state, executes an action from a given  alphabet $\Sigma$ and ends in a target distribution~\cite{Seg95}. Such a distribution is then resolved into a probabilistic choice which specifies the new state of the automaton up to some probabilistic factor.

\begin{definition}\label{def:probabilistic-automaton}
A probabilistic automaton is a tuple $(P,\Sigma,\lra,\phi_P,F_P)$ where 
\begin{itemize}
\item[-] $P$ is a set of states,
\item[-] $\Sigma$ is a set of actions,
\item[-] $\lra:P\times\Sigma\times\DD(P)$ is a set of probabilistic transitions where $\DD(P)$ is the set of finitely (or countably) supported probability distributions over $P$,
\item[-] $\phi_P$ is the start or initial distribution over the states in $P$,
\item[-] and $F_P$ is a set of final states.
\end{itemize}
\end{definition}

We usually identify an automoton with its set of states and explicit distinction will be made only when confusion may arise.
\begin{example}
Figure~\ref{fig:vending-machine} depicts two probabilistic automata. The automaton on the left models a faulty vending machine that becomes stuck with probability $0.2$; on the right the automaton represents the actions of a user interacting with the automaton by kicking it if he fails to get his tea. Two kicks means the machine really is broken. \footnote{This example was suggested by Steve Schneider~\cite{Ste12}.} 

The states of the two automata are labelled by $s_i,t_i$ respectively and distributions are not labelled unless they are initial and their components correspond to dotted arrows labelled with the probability. The set of actions is $\Sigma = \{\coin,\tea,\kick,\fail,\stuck\}$ where $\stuck$ is the only internal action. We assume that the two automata have no final state to facilitate the upcoming calculations. 
\begin{figure}[!h]
\begin{displaymath}
\xymatrix{&s_0\ar[d]^{\coin}& \\
&\arp[dl]|{0.2} \arp[dr]|{0.8}& \\
s_1\ar[d]_{\stuck}& &s_2\ar[d]^{\tea}\\
s_3\ar@/_/[uur]_{\kick}&&s_4
}\qquad\qquad
\xymatrix{
 &t_1\ar[d]_{\coin} &\\
 &\ar[dl]_{\kick}t_2\ar[dr]^{\tea}\\
 t_3\ar[d]_{\kick}\ar[rr]_{\tea}&&t_4\\
t_5\ar@(d,r)_{\fail}&&
}
\end{displaymath}
\caption{A probabilistic vending machine $V = \coin\cdot M$ and a user $U = \coin\cdot U'$ who kicks the machine once if it gets stuck.} \label{fig:vending-machine}
\end{figure}
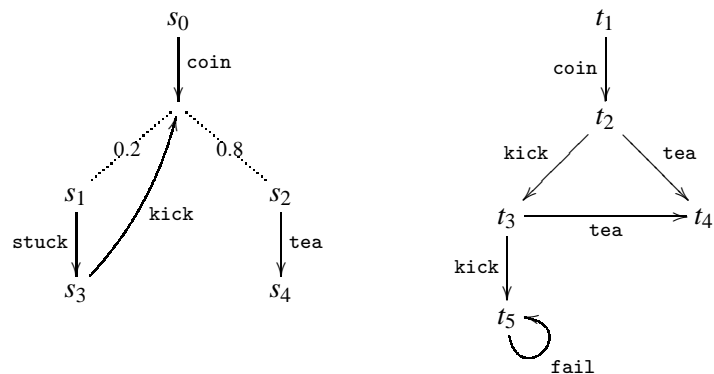
\end{example}

Definition~\ref{def:probabilistic-automaton} provides a specialised version of probabilistic automata. Generally, a transition is composed of a state and a distribution over $\Sigma\times P$ but restriction to automata with simple transitions suffices for our results in this paper. We denote by $\pAut$ the set of such probabilistic automata.

The set of actions $\Sigma$ is divided into two parts, namely, \textit{internal} and \textit{external} actions. Internal actions are either local or invisible and are usually intrinsic to the automaton where they are defined. They are not shared with other automata in the sense that they can be executed independently from the environment. A special case is the silent action $\tau$ which does not belong to the set of internal actions $I$ and we write $I_\tau = I\cup\{\tau\}$. In contrast, external actions are visible to the environment and may be synchronised. We denote the set of external actions by $E$, and define $\Sigma=I\cup E$ and $\Sigma_\tau = I_\tau\cup E$. The set $\Sigma$ is assumed implicitly and is fixed for every automaton.

The linear run of a probabilistic automaton yields a \textit{path}, as in the standard case, which is quantified with respect to a family of probability measures (indexed with a set of probabilistic scheduler). Formally, a path is a sequence $x_0a_1x_1a_2x_2\cdots$ of alternating states and actions such that there is a sequence of transitions $\trans{x_i}{a_{i+1}}{\mu_{i+1}}$, $i\geq 0$, where $x_i\in\supp(\mu_i)$ (the support of $\mu$) for every $i>0$. A path $\alpha$ always starts with a state and ends with another state, denoted $\last(\alpha)$, if it is finite. Usually, we want a path to start from a state in the support of the initial distribution. We denote $\Path(P)$ the set of all finite paths of an automaton $P$.

Next we provide some operations over probabilistic automata. The regular operators include nondeterministic choice ($+$), sequential composition ($\cdot$) and Kleene star ($*$) which abstracts tail iteration together with the constants skip ($\one$),  deadlock ($\zero$) and the automaton that enables a single successfully terminating action. Such an automaton is denoted $\mathbf{a}$ where the action is $a$. Formally, we have

\begin{itemize}
\item[-] Deadlock $\zero$ corresponds to $(\{x\},\emptyset,\delta_x,\emptyset)$ where $\delta_x$ is the point probability distribution concentrated on the state $x$.
\item[-] Skip $\one$ corresponds to $(\{x\},\emptyset,\delta_x,\{x\})$.
\item[-] $\Pa$ corresponds to $(\{x,x'\},\{\trans{x}{a}{\delta_{x'}}\},\delta_x,\{x'\})$.
\end{itemize}

In the reminder of this section, we fix two probabilistic automata $P,Q$ with respective initial distribution $\mu_0,\nu_0$, sets of final states $F_P,F_Q$ and sets of transitions $\lra_P,\lra_Q$. We also assume that the state spaces of $P$ and $Q$ are disjoint. We now give an automata semantics for each of the named operators.

Nondeterminism is defined as in the standard case by constructing a new initial point distribution $\delta_z$ such that $z$ leads to the respective initial distributions of the operands via $\tau$ transitions.

\begin{definition}
Nondeterministic choice between $P$ and $Q$ is defined by:
$$P+Q = (P\cup Q\cup\{z\},\lra_P\cup\lra_Q\cup\{\trans{z}{\tau}{\mu_0},\trans{z}{\tau}{\nu_0}\},\delta_{z},F_P\cup F_Q)$$ where $z\notin P\cup Q$.
\end{definition}

\begin{definition}
The sequential composition of $P$ followed by $Q$ is defined by: $$P\cdot Q = (P\cup Q,\lra_{P\cdot Q},\mu_0,F_Q)$$ where 
$$\lra_{P\cdot Q} = \lra_P\cup\lra_Q\cup\{\trans{x}{\tau}{\nu_0}\ |\ x\in F_P \}.$$
\end{definition}

This definition is a straightforward generalisation of the standard definition from automata theory.

Now we define the Kleene star in the standard way.

\begin{definition}
The tail iteration or Kleene star of $P$ is $P^* = (P\cup \{z\},\lra_{P^*},\delta_z,\{z\})$ where 
$$\lra_{P^*} = \lra_P\cup\{\trans{z}{\tau}{\mu_0},\trans{x}{\tau}{\delta_z}\ |\ x\in F_P\}$$
and $z\notin P$.
\end{definition}

The implementation of a probabilistic choice between two automata is defined below. 
\begin{definition}
We define the probabilistic choice between $P$, with probability $p$, and $Q$, with probability $1-p$, as
$$P\pc{p} Q = (P\cup Q,\lra_P\cup\lra_Q,p\mu_0+(1-p)\nu_0,F_P\cup F_Q).$$
\end{definition}

Finally, the parallel composition is defined using a probabilistic version of CSP parallel composition operation that synchronises on the actions in $A\subseteq E$. The frame set $A$ is assumed to be fixed throughout this paper.

Firstly, given $\mu\in\DD(P) $ and $\nu\in\DD(Q)$, the product $\mu\times\nu$ is a distribution over $P\times Q$ such that $(\mu\times\nu)(x,y) = \mu(x)\nu(y)$ (component-wise multiplication).
\begin{definition}
We define the parallel composition of $P$ and $Q$ as 
$$P\pr{A} Q = (P\times Q,\lra_{P\pr{A} Q},\mu_0\times\nu_0,F_P\times F_Q)$$ where, for each $a\in\Sigma_\tau$, a transition $\trans{(x,y)}{a}{\mu\times\nu}$ belongs to $\lra_{P\|Q}$ if one of the following conditions holds:
\begin{itemize}
\item $a\in A$ and $\trans{x}{a}{\mu}$ and $\trans{y}{a}{\nu}$,
\item $a\notin A$ and $\trans{x}{a}{\mu}$ and $\nu=\delta_y$,
\item $a\notin A$ and $\trans{y}{a}{\nu}$ and $\mu = \delta_x$.
\end{itemize}
\end{definition}

In the construction of the transition relations of $P\pr{A}Q$, given a transition of $P$, if it is labelled by an action in $A$ then it is blocked until it gets synchronised with a transition of $Q$ labelled with the same action. Otherwise, that transition is interleaved with the transitions of $Q$.

\begin{example}~\label{ex:algebra}
Using this language we can describe the automata from Figure~\ref{fig:vending-machine}. The right hand side automaton of Figure~\ref{fig:vending-machine} corresponds to the algebraic expression 
$$\coin\cdot(\kick\cdot(\kick\cdot\fail^* + \tea) + \tea)$$
where we have abused notation by denoting the automaton that does a single action, say $\coin$, and then terminates successfully with the same notation $\coin$.

The left hand side is obtained as a sequential composition $\coin\cdot M$ where $M$ corresponds to the least fixed point of 
$$f(X) = \stuck\cdot\kick\cdot X\cdot \mathbf{0}\ \pc{0.2} \tea\cdot\mathbf{0}.$$
We will express the least fixed point of $f$ as an algebraic expression in Section~\ref{sec:soundness}. 
\end{example}

\section{Probabilistic Simulation}\label{sec:simulation}

In this section, we define an inequality on the set $\pAut$ as per the constructions of~\cite{Den07b,Den07a,Seg03,Seg96}. The equivalence relation is based on weak simulation and we are mainly interested in the equational theory of this model and relate it to the axiomatisation of probabilistic and concurrent Kleene algebras.

We provide two equivalent definitions of simulation. The first definition is the probabilistic simulation of~\cite{Den07a} and the second is probabilistic weak forward simulation of~\cite{Seg96}. Both definitions of simulation rely on the lifting of relations from states to distributions. Given a relation $S\subseteq X\times\DD(Y)$, the lifting~\cite{Den07a} of $S$ is a relation  $\ov{S}\subseteq \DD(X)\times\DD(Y)$ such that $\mu\ov{S}\nu$ if and only if there exists a family of real number $\{p_n\ |\ n\in N\}\subseteq[0,1]$ such that $\sum_np_n = 1$ and
\begin{enumerate}
\item $\mu=\sum_{n\in N}p_n\delta_{x_n}$,
\item for each $n\in N$, there exists $\nu_n\in\DD(X)$ such that $x_nS\nu_n$,
\item $\nu = \sum_{n\in N}p_n\nu_n$.
\end{enumerate}

The lifting is a probability preserving function that associates to each probabilistic relation $R$ a standard relation $\ov{R}$ over the set of distributions. It is important to notice that the decomposition of $\mu$ is not necessarily canonical that is, there may be some repetition in the $x_i$s. Moreover, the lifting also applies to labelled transition because $\trans{\cdot}{a}{\cdot}\subseteq P\times\DD(P)$ for any probabilistic automaton $P$ and any action $a\in\Sigma_\tau$. Hence, we denote $\trans{}{\ov{a}}{}$ the lifting of this transition which corresponds to the notion of combined transition of~\cite{Seg03,Seg96}.

Lastly, we extend internal transitions with reflexivity, that is, we write $\trans{x}{\tau}{\mu}$ if such a transition exists in the automaton or $\mu = \delta_x$. The lifted version is again denoted $\trans{}{\ov{\tau}}{}\subseteq\DD(P)\times\DD(P)$. Finally, weak transitions are obtained from the reflexive transitive closure of $\trans{}{\ov{\tau}}{}$, denoted $\Trans{}{}{}$, and we write $\Trans{\mu}{a}{\mu'}$ if there exist $\mu_1,\mu_2$ such that $\Trans{\mu}{}{\trans{\mu_1}{\ov{a}}{\Trans{\mu_2}{}{\mu'}}}$.

We now give the formal definition of simulation by straightforwardly generalising~\cite{Den07a} to automata with final states.
\begin{definition}\label{df:simulation}
A probabilistic simulation $S$ from $P$ to $Q$ is a relation $S\subseteq P\times \DD(Q)$ satisfying the following properties:
\begin{enumerate}
\item[1.] there exists $\nu_0\rq{}$ such that $\mu_0 \ov{S}\nu_0\rq{}$ and $\Trans{\nu_0}{}{\nu_0\rq{}}$,
\item[2.] if $\trans{x}{a}{\mu\rq{}}$ is a valid transition of $P$ and $xS\nu$, there exists $\nu\rq{}\in\DD(Q)$ such that $\Trans{\nu}{a}{\nu\rq{}}$ and $\mu\rq{}\ov{S}\nu\rq{}$,
\item[3.] if $x\in F_P$ and $xS\nu$ then there exists $\nu'\in\DD(F_Q)$ such that $\Trans{\nu}{}{\nu'}$. 
\end{enumerate}
\end{definition}

Property (a) ensures that preceding $\tau$ actions do not interfere with probabilistic choice (i.e. $P\pc{p} Q$ and $\tau\cdot(P\pc{p}Q)$ are equal). Property (b) is the usual co-inductive definition of simulation and property (c) ensures that if a state $x\in P$ is simulated by a distribution $\nu\in\DD(Q)$ and $P$ can terminate successfully at $x$ then $Q$ can also terminate successfully from $\nu$ after a \textit{finite number} of internal transitions.

A simulation is always total on reachable states, that is, if $S\subseteq P\times\DD(Q)$ is a simulation and $x\in P$ such that $x_0a_1x_1\cdots x$ is a path that occurs with positive maximal probability, then there exists $\nu\in\DD(Q)$ such that $x S\nu$.

We write $P\leq Q$ if there is a simulation from $P$ to $Q$ and $P\equiv Q$ iff  $P\leq Q$ and $Q\leq P$.

\begin{example}
Figure~\ref{fig:sim} depicts two automata related by a simulation relation i.e. $M\leq H$ where $M$ (resp. $H$) is the left (resp. right) automaton. The simulation is obtained from the relation $S = S'\cup\{(s_3,\mu)\ |\ (s_1,\mu)\in S'\}$ where 
$$S'=\{(s_1,0.2\delta_{u_0}+0.8\delta_{u_1}),(s_1,\delta_{u_2}),(s_1,\delta_{u_4}),(s_2,\delta_{u_1}),(s_2,\delta_{u_3}),(s_4,\delta_{u_5})\}.$$
In fact, we can write $\nu_0=0.2(0.2\delta_{u_0}+0.8\delta_{u_1})+0.8\delta_{u_1}$ where $s_1S(0.2\delta_{u_0} + 0.8\delta_{u_1})$ and $s_2S\delta_{u_1}$. Hence, $\mu_0\ov S\nu_0$. Since $\stuck$ is an internal action, it follows $s_3S\mu$ and $\Trans{\mu}{\ov{\tau}}{\mu}$ for every distribution $\mu$ such that $s_1S\mu$. Next, we have $s_3S(0.2\delta_{u_0}+0.8\delta_{u_1})$ and $\trans{s_3}{\kick}{\mu_0}$. Since $\mu_0=0.2\delta_{s_1}+0.8\delta_{s_2}$ and $s_1S\delta_{u_2}$ and $s_2S\delta_{u_3}$, it follows that $\mu_0\ov{S} (0.2\delta_{u_2}+0.8\delta_{u_3})$ and $\Trans{(0.2\delta_{u_0}+0.8\delta_{u_1})}{\ov{\kick}}{(0.2\delta_{u_2}+0.8\delta_{u_3})}$. The other inductive cases are proved in similar fashion. Moreover, an algebraic proof is given in the next section.
\begin{figure}
\begin{displaymath}
\xymatrix{
&\arp[dl]|{0.2}\mu_0 \arp[dr]|{0.8}&  &\hspace{2cm}&&\nu_0\arp[dl]|{0.04}\arp[dr]|{0.96}& \\
s_1\ar@/_/[drrrr]\ar@/_/[ddrrrr]\ar@/_/@{.>}[rrrr]|<<<<<<<<{0.2}\ar@/^/@{.>}[rrrrrr]|<<<<<<<{0.8}\ar[dd]_{\stuck}& &s_2\ar@/_/[rrrd]\ar@/_/[rrrr]\ar[dd]_>>>>>>\tea &&u_0\ar[d]_{\kick}& &u_1\ar[dl]_{\kick}\ar[dd]^{\tea}\\
&& &&u_2\ar[d]_{\kick}& u_3\ar[rd]_{\tea} & \\
s_3\ar@/_/[uuur]_<<<<<<\kick&&s_4\ar@/_/[rrrr]&&u_4\ar@(d,l)^{\kick,\tea,\fail}&&u_5
}
\end{displaymath}
\caption{Two automata related by simulation: the left is M and the right is H. Remind that $\stuck$ is an internal action so we have removed the arrows from $s_3$ because they are exactly the same as for $s_1$. The dotted arrow represents non-trivial distribution again.}\label{fig:sim}
\end{figure}
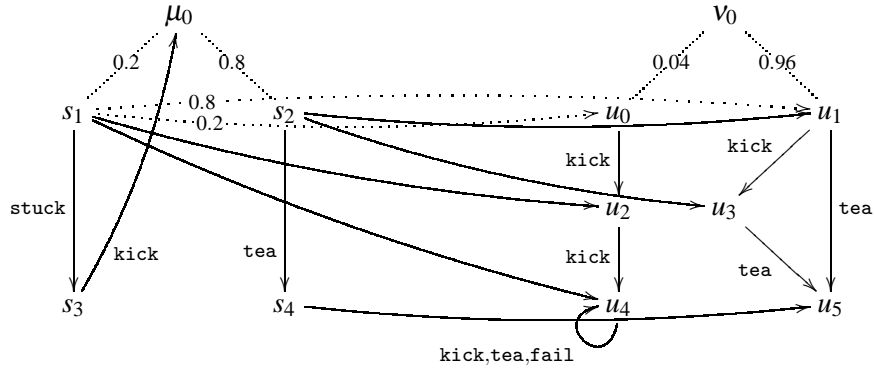
\end{example}
\begin{proposition}\label{pro:sim-partial-order}
The simulation relation is a preorder.
\end{proposition}

Any probabilistic simulation satisfying Definition~\ref{df:simulation} will be referred simply as a simulation. In contrast, the definition of forward simulation~\cite{Seg96} relies on a \textit{double lifting}. Given a relation $S\subseteq X\times Y$, the double lifting of $S$, denoted $\ovov{S}$, is a subset of $\DD(X)\times\DD(Y)$ where $\mu\ov{\ov{S}}\nu$ iff there exists a function $w:X\times Y\to[0,1]$ such that 
\begin{enumerate}
\item if $w(x,y)>0$ then $x Sy$,
\item for every $x\in X$, $\sum_{y\in Y}w(x,y) = \mu(x)$,
\item for every $y\in Y$, $\sum_{x\in X}w(x,y) = \nu(y)$.
\end{enumerate}

The function $w$ is again a probability preserving function that provides corresponding decompositions for $\mu$ and $\nu$. Double lifting generates a distribution over the set of distributions which complicates the lifting of transitions. To obtain a standard relation over the set of distributions, Segala~\cite{Seg03,Seg96} provided a flat version of a distribution in $\DD(\DD(Q))$ through the use of $\pi:\DD(\DD(X))\to\DD(X)$ such that 
$$\pi(\phi) = \sum_{\mu\in\supp(\phi)}\phi(\mu)\mu.$$
We now give the modified version of Segala's probabilistic weak forward simulation.
\begin{definition}\label{df:pwfs}
A relation $S\subseteq P\times \DD(Q)$ is a probabilistic weak forward simulation if 
\begin{itemize}
\item[a)] there exist $\psi_0\in\DD(\DD(Q))$, such that $\mu_0 \ovov{S}\psi_0$ and $\Trans{\nu_0}{}{\pi(\psi_0)}$,
\item[b)] if $\trans{x}{a}{\mu\rq{}}$ is a valid transition of $P$ and $xS\nu$, there exists $\psi\in\DD(\DD(Q))$ such that $\Trans{\nu}{a}{\pi(\psi)}$ and $\mu\rq{}\ov{\ov{S}}\psi$.
\item[c)] if $x\in F_P$ and $xS\nu$ then there exists $\psi\in\DD(\DD(F_Q))$ such that $\Trans{\nu}{}{\pi(\psi)}$.
\end{itemize}
\end{definition}

\begin{proposition}
Let $X,Y$ be two sets, $S\subseteq X\times\DD(Y)$, $\mu\in\DD(X)$ and $\psi\in\DD(\DD(Y))$. If $\mu\ovov{S}\psi$ then $\mu\ov{S}\pi(\psi)$.
\end{proposition}
\begin{proof}
If $\mu\ovov{S}\psi$, then there exists $w:X\times\DD(Y)\to[0,1]$ satisfying the condition above. Then by considering $I = \supp(w)$, it directly follows that $\mu = \sum_{i\in I} w(i)\delta_{x_i}$, each $x_i$ is related to some $\nu_i$ and $\pi(\psi) = \sum_{i\in I}w(i)\nu_i$.
\end{proof}

\begin{corollary}\label{cor:equivalence-of-simulation}
A relation is a probabilistic simulation iff it is a probabilistic weak forward simulation on $\pAut$.
\end{corollary}

\begin{proof}
We provide a sketch for the proof though the complete proof should not be hard to obtain from it.

That Definition~\ref{df:pwfs} implies Definition~\ref{df:simulation} follows directly from the previous proposition.

Assume that $S\subseteq P\times\DD(Q)$ satisfies Definition~\ref{df:simulation}. If $\mu\ov{S}\nu$, then there exsits a decomposition $\mu = \sum_{i\in I}p_i\delta_{x_i}$ such that for each $i$, there exists $\nu_i\in\DD(Q)$ such that $x_i\ov{S}\nu_i$ for each $i$, and $\nu = \sum_{i\in I}p_i\nu_i$. Hence $\mu\ovov{S}\sum_{i\in I}p_i\delta_{\nu_i}$. We just apply this simple construction for each of the three cases.
\end{proof}

\begin{proposition}
Simulation is a precongruence i.e. if $P\leq Q$ then $P+R\leq Q+R$, $P\cdot R\leq Q\cdot R$,  $P^*\leq Q^*$, $P\pc{p}R\leq Q\pc{p}R$, $P\pr{A} R\leq Q\pr{A} R$ and the same holds for binary operators when the order of the arguments is reversed.
\end{proposition}

The proof of this proposition can be found in the Appendix Proposition~\ref{ap:pro:precongruence}.

We conclude this section with a remark about the simulation of Definitions~\ref{df:simulation} and~\ref{df:pwfs}. In Proposition~\ref{cor:equivalence-of-simulation}, we have shown that the corresponding definitions of~\cite{Seg95} and~\cite{Den07a} coincide (Notice that we can replace final states with some special external action and obtain a formulation closer to~\cite{Den07a,Seg96}). On the one hand, Segala has shown that the largest precongruence included in the trace distribution equivalence coincides with \textit{vector may testing} where there are uncountably many success actions~\cite{Seg96}. On the other hand, Deng et al. have shown that vector and scalar testings coincide on the recursion-free fragment of probabilistic automata and that with the same restriction, Definition~\ref{df:simulation} is complete for testing equivalence~\cite{Den07b}. With the help of Proposition~\ref{cor:equivalence-of-simulation} and the equivalence between probabilistic weak forward simulation (Definition~\ref{df:pwfs}) and the coarsest precongruence included in the trace distribution, we conclude that Deng's completeness for may testing extends to automata with countable state spaces. However, it is still unknown whether the equivalence between scalar and vector testing in the infinite case is valid. 

These equivalences are the main motivation for our use of simulation in order to create an interleaving model for our algebra. It should be noted that simulation equivalence is decidable for finite automata but it is unknown whether an efficient decision procedure exists. This is in contrast to other related results in the literature showing that strong simulation is decidable in polynomial time~\cite{Kom12}.

\section{Probabilistic Concurrent Kleene Algebra}\label{sec:soundness}

In this section, we introduce probabilistic concurrent Kleene algebra. We show that the set of probabilistic automata modulo probabilistic simulation as defined by Definitions~\ref{df:simulation} or~\ref{df:pwfs} satisfies an extension of Kleene algebras that includes probability and concurrency. 

Concurrent Kleene Algebra has four unary and binary operators, namely, $+,\cdot,\|$ and $*$. These operators have pomset operation semantics but the axiomatisation is too weak to allow the presence of probability. In contrast, Probabilistic Kleene Algebra has three operators, namely, $+,\cdot$ and $*$. We ensure that these operators coincide with the respective operators of Concurrent Kleene Algebra and provide a new probabilistic concurrent Kleene algebra that extends both structures. Without explicit probabilistic choice, such combination generates a weak Concurrent Kleene Algebra.

\begin{definition}
A weak concurrent Kleene algebra is an algebraic structure with signature $(K,+,\cdot,\|,*,\zero,\one)$ where $K$ is a set closed under the operations and satisfies  Equations~(\ref{eq:+-idem}-\ref{eq:+-assoc}),~(\ref{eq:right-skip}-\ref{eq:+-left-subdist}) and (\ref{eq:*-unfold}-\ref{eq:|-subdistributes-+})
\end{definition}

To gain complete control of probabilities, we append explicit probabilistic choices to weak concurrent Kleene algebra. Many of the following equations have been proven elsewhere so we will prove only those that are specific to our algebra. We concentrate on equations describing the interactions between probabilistic choices, sequential composition, Kleene star and the exchange law~(\ref{eq:exchange-law}).

We assume the following precedence between the operators. The Kleene star $*$ binds more tightly than $\cdot$ which binds more tightly than $\|$. The operator $\|$ binds more tightly than $+$ and $\pc{p}$ and we use parenthesis to parse expressions having $+$ and $\pc{p}$ at the same level.

The following equations are standard and the proofs are omitted (they can be found in ~\cite{Den07a}).

\begin{eqnarray}
P & \equiv & P+P\label{eq:+-idem}\\
P & \equiv & P+ \zero\\
P+Q & \equiv & Q+P \\
P+(Q+R)&\equiv& (P+Q)+R\label{eq:+-assoc}\\
P & \equiv & P\pc{p}P\label{eq:pc-idem}\\
P\pc{p}Q & \equiv & Q\pc{1-p } P \\
P\pc{p}(Q\pc{q}R)&\equiv& (P\pc{p\rq{}}Q)\pc{q\rq{}}R\label{eq:pc-assoc}
\end{eqnarray}
where $p\rq{}q\rq{}=p$, $(1-p\rq{})q\rq{} = (1-p)q$ and $1-q\rq{} = (1-p)(1-q)$. Moreover, the equivalence $P\leq Q$ iff $P+Q\equiv Q$ follows from these equations, that is, simulation coincides with the natural order of the algebra. Remind that in our interpretation $Q$ has more behaviours than $P$. A complete characterisation of the consequences of Equation~(\ref{eq:pc-idem}-\ref{eq:pc-assoc}) with respect to probabilistic bisimulation can be found in~\cite{Sta96}.

The proof of the following propositions can be found in the Appendix under Proposition~\ref{ap:pro:sequential},~\ref{ap:pro:kleene-star} and~\ref{ap:pro:parallel} respectively.

\begin{proposition}
The sequential composition satisfies 
\begin{eqnarray}
P & \equiv & P\cdot\one\label{eq:right-skip}\\
P & \equiv & \one\cdot P\label{eq:left-skip}\\
\zero & \equiv & \zero\cdot P\label{eq:left-zero}\\
P\cdot(Q\cdot R)&\equiv& (P\cdot Q)\cdot R\label{eq:.-assoc}\\
P\cdot R + Q\cdot R&\equiv& (P + Q)\cdot R\label{eq:+-right-dist}\\
P\cdot Q + P\cdot R&\leq& P\cdot(Q+ R)\label{eq:+-left-subdist}\\
(P\pc{p} Q)\cdot R&\equiv&P\cdot R \pc{p} Q\cdot R\label{eq:pc-dist}\\
 P\cdot(Q\pc{p} R)&\leq&P\cdot Q \pc{p} P\cdot R\label{eq:pc-supdist}
\end{eqnarray}
\end{proposition}

\begin{proposition}
The Kleene star satisfies the following laws:
\begin{eqnarray}
P^* & \equiv & \one + P\cdot P^*\label{eq:*-unfold}\\
P\cdot Q\leq Q & \Rightarrow & P^*\cdot Q\leq Q\label{eq:*-induction}
\end{eqnarray}
\end{proposition}

For the parallel composition, we assume synchronisation over all external actions and denote it simply with $\|$ without any frame as defined in Section~\ref{sec:p-automata}.
\begin{proposition}
The parallel composition satisfies
\begin{eqnarray}
P\|Q & \equiv & Q\|P \label{eq:|-comm}\\
P\|(Q\|R)&\equiv& (P\|Q)\|R\label{eq:|-assoc}\\
(P\|Q)\cdot (P\rq{}\|Q\rq{})&\leq& P\cdot P\rq{}\|Q\cdot Q\rq{}\label{eq:exchange-law}\\
P\| Q + P\| R&\leq& P\|(Q +R)\label{eq:|-subdistributes-+}\\
P\| Q \pc{p} P\| R&\equiv& P\|(Q \pc{p}R)\label{eq:|-distributes-pc}
\end{eqnarray}
\end{proposition}

Notice that we cannot have equality for the interchange law~(\ref{eq:exchange-law}) even with a fully synchronised $\|$. For example $\Pa\|\Pa \equiv (\Pa\cdot \one)\|(\one\cdot\Pa) > (\Pa\|\one)\cdot (\one\|\Pa) \equiv \zero$ where we assume that the action $a$ is external, hence sychronised.

\begin{theorem}
$(\pAut,+,\cdot,\|,^*,\Deadlock,\Skip)$ is a weak concurrent Kleene algebra.
\end{theorem}

\begin{definition}
A probabilistic concurrent Kleene algebra is a weak concurrent Kleene algebra with a collection of probabilistic choices $\pc{p}$, $p\in[0,1]$, satisfying equations~(\ref{eq:pc-idem}-\ref{eq:pc-assoc},\ref{eq:pc-dist}-\ref{eq:pc-supdist},\ref{eq:|-distributes-pc}).
\end{definition}

\begin{example}
We end this section by providing an algebraic proof for the existence of a simulation between the automata in Figure~\ref{fig:sim}. First, we express the least fixed point of the function $f$ of Example~\ref{ex:algebra} as promised. We prove that
$$f(X) = (\stuck\cdot\kick\pc{0.2} \tea\cdot\mathbf{0})\cdot X\cdot\zero$$ 
using equations~(\ref{eq:left-zero}) and~(\ref{eq:pc-dist}). Now, we show that the least fixed point of $f(X) = P\cdot X\cdot \zero$ is $P^*\cdot \zero$. In fact $f(P^*\cdot \zero) = P\cdot P^*\cdot \zero = (\one + P\cdot P^*)\cdot \zero = P^*\cdot \zero$ because of equations~(\ref{eq:left-skip}), ~(\ref{eq:+-right-dist}) and~(\ref{eq:*-unfold}). Now let $Q$ be a suffix point of $f$ i.e. $P\cdot Q\cdot \zero\leq Q$, then monotonicity and Equation~(\ref{eq:left-zero}) implies $P\cdot Q\cdot \zero\leq Q\cdot \zero$. Therefore, $P^*\cdot Q\cdot \zero\leq Q\cdot \zero\leq Q$ because of the induction law~(\ref{eq:*-induction}) and $\zero\leq\one$. Hence $P^*\cdot\zero\leq Q$ follows from Equation~(\ref{eq:left-zero}) and monotonicity of $\cdot$.

Therefore, the left hand side automaton is simulation equivalent to 
$$M = (\stuck\cdot\kick\pc{0.2}\tea\cdot \zero)^*\cdot \zero$$
One unfold of this automaton gives 
\begin{eqnarray}
M & \equiv & (\stuck\cdot\kick\pc{0.2}\tea\cdot \zero)\cdot (\stuck\cdot\kick\pc{0.2}\tea\cdot \zero)^*\cdot \zero\nonumber \\
&\equiv & \stuck\cdot\kick\cdot(\stuck\cdot\kick\pc{0.2} \tea\cdot\zero)^*\cdot\zero\pc{0.2}\tea\cdot\zero\nonumber\\
&\equiv & \stuck\cdot\kick\cdot(\stuck\cdot\kick\pc{0.2}\tea\cdot\zero)\cdot M\pc{0.2}\tea\cdot\zero\nonumber\\
&\leq & (\stuck\cdot\kick\cdot\stuck\cdot\kick\cdot M\pc{0.2}\stuck\cdot\kick\cdot\tea\cdot\zero)\pc{0.2}\tea\cdot\zero\nonumber\\
& \equiv & \stuck\cdot\kick\cdot\stuck\cdot\kick\cdot M\pc{0.04}(\stuck\cdot\kick\cdot\tea\cdot\zero\pc{0.16/0.96}\tea\cdot\zero)\nonumber\\
&\leq&\stuck\cdot\kick\cdot\stuck\cdot\kick\cdot M\pc{0.04}(\stuck\cdot\kick\cdot\tea\cdot\zero + \tea\cdot\zero)\nonumber
\end{eqnarray}
The second equality follows from Equations~(\ref{eq:pc-dist}) and~(\ref{eq:left-zero}). The third equality follows from an unfolding of the Kleene star and the definition of M. The fourth inequality follows from Equation~(\ref{eq:pc-supdist}). The fifth equality follows from Equation~(\ref{eq:pc-assoc}) and in the last equality, we have used the fact that $P\pc{p}Q\leq P+Q$. We use monotonicity to finally deduce that $M\leq H$ because $M\leq\run(\{\kick,\tea,\fail\})$ where $\run(\{a_i\}_{i=0,n})) = (\sum_{i=0}^n a_i)^*$ and 
$$H = \kick\cdot\kick\cdot\run(\{\kick,\tea,\fail\})\pc{0.04}(\kick\cdot\tea\cdot\zero + \tea\cdot\zero).$$
\end{example}

\section{Rely-Guarantee Rules}~\label{sec:rely-guarantee}

The rely/guarantee formalism provides a powerful tool for verifying a system with multiple interacting components. The concept is based on deriving the properties of the larger system through the use of inference rules on the specification of the components. We are interested in generating these inference rules algebraically.

In this section we extend the algebraic formulation of the rely-guarantee calculus in~\cite{Hoa09} to include probabilistic systems. Notice that the automata model provides a particular interpretation of rely/guarantee tuples because it is action based though the algebraic laws persist to any model satisfying the axioms. Moreover, our results are based on the definition of~\cite{Hoa09} instead of providing another interpretation of Jones' rely/guarantee components~\cite{Jon81}. 

A rely/guarantee quintuple is composed of five components $P\ R\ \{U\}\ Q\ G$ where $P,Q$ are pre/postconditions, $R$ is a rely property and $G$ is the guaranteed part. These components are usually algebraic specifications and can be interpreted as automata when needed where the parallel composition has a fixed frame. In the automata model, $\run$ denotes the automata with self-loop constructed from the external actions and the algebraic proofs are valid for every model where $\run$ is replaced by the unit of $\|$ (if it exists).

\begin{definition}
A rely guarantee quintuple $P\ R\ \{U\}\ Q\ G$ holds if and only if 
$$P\cdot(R\|U)\leq Q\quad \wedge\quad U\leq G$$
\end{definition}

In other words, if $U$ is part of a system that satisfies the property $R$, then the system will satisfy the specification determined by $(P,Q)$ and $U$'s behaviour is determined by $G$.

The difference between our approach and that of Kwiatkowska et al in~\cite{Kwi10} lies in the definition of the rely/guarantee tuple, that is, the interpretations of \textrm{tuple satisfaction} are different. In~\cite{Kwi10}, the guaranteed property is somehow part of the postcondition and they are expressed using safety properties instead of simulation relations. Our approach is an extension of the work in~\cite{Hoa09} and has a similar flavour as~\cite{Kom12}.

We now provide some simple rely/guarantee rules together with their algebraic proofs. Notice that these rules are valid for all models satisfying the axioms of probabilistic concurrent Kleene algebra. This nicely illustrates the power of algebras where they provide results that are model independent. Of course, this power comes with the disadvantage that many details are lost through abstraction. The importance of these details depends on the system and the properties to be studied.

Remind that, given a finite set of external actions $A = \{a_1,\dots,a_n\}$, we denote $\run(A) = (\sum_{i=1}^na_i)^*$ which is similar to the run of standard CSP. Notice that for every term $P$ constructed from $A$ and the algebraic operators, we have $P\leq\run(A)$ and $P\pr{A}\run(A) = P$.

An \emph{isolated system} is composed of two components that interact without any interference from an outside environment. An example of isolated system is given by our vending machine and user. 

\begin{proposition}
For every term $P,P',Q,Q',U,U',R,R',G$ and $G'$ we have the following concurrent rule for isolated system:

\begin{equation}\label{rul:sum}
\frac{P\ R\ \{U\}\ Q\ G\qquad  P'\ R'\ \{U'\}\ Q'\ G'\qquad  G\leq R'\qquad  G'\leq R}{T\ \run\ \{U\|U'\}\ Q\ (G\|G')}
\end{equation}

where $\run$ is constructed from the external actions of $U\|U'$ and the conclusion is valid whenever $T \leq P$ and $T\leq P\rq{}$.
\end{proposition}

\begin{proof}
We have $P\cdot (R\|U)\leq Q$ and $U\leq G$ from the interpretation of the first quintuple. The premise $G\leq R'$ implies that $U\leq R'$. Therefore, monotonicity yields 
$P\cdot (U\|U')\leq Q$. Hence, if $T\leq P$ and $T\leq P'$ then $T\cdot(U\|U')$ satisfies both $Q$. Since we assume that the external actions of $U$ and $U'$ coincide which is used in the parametrisation of $\|$ and $\leq$, we have $\run\|(U\|U') = U\|U'$ and we obtain the guaranteed part of the conclusion with monotonicity of $\|$.
\end{proof}

The Rule~(\ref{rul:sum}) implies that the quintuple $T\ \run\ \{U\|U'\}\ Q'\ (G\|G')$ also holds. This rule can only be used for isolated systems.

\begin{corollary}
The following asymmetric rule holds for isolated systems
\begin{equation}\label{rul:asym}
\frac{\one\ \run\ \{U\}\ \run\ G\ \qquad\ P'\ R'\ \{U'\}\ Q'\ G'\ \qquad\ G\leq R'}{P'\ \run\ \{U\|U'\}\ Q'\ (G\|G')}
\end{equation}
\end{corollary}
When the system is not isolated, a more general rule is needed. We can show that if there exists $S$ such that $S\leq R$, $S\leq R'$ and $S\|S\leq S$ then we can infer from the premises of Rule~(\ref{rul:sum}) that the quintuple
$$T\ S\ \{U\|U'\}\ Q\ (G\|G')$$ holds, where $T\leq P$ and $T\leq P'$. In other words, if the system $U\|U'$ is run within an environment that guarantees $S$ then it satisfies the postconditions $Q$ and $Q'$ and guarantees $G\|G'$.

\begin{proposition}
We have the following sequential rule:
\begin{eqnarray}
\frac{\begin{array}{c}
P\ R\ \{U\}\ Q\ G\ \qquad\ P'\ R'\ \{U'\}\ Q'\ G'\ \qquad\ Q\leq P' \\ 
\qquad\ (R\|U)\cdot(R'\|U') = (R\cdot R')\|(U\cdot U')
\end{array}}{
\begin{array}{c}P\ (R\cdot R')\ \{U\cdot U'\}\ Q'\ (G\cdot G')
\end{array}}\label{rul:seq}
\end{eqnarray}
\end{proposition}

\begin{proof}
We have $P\cdot(R\|U)\leq Q$ and $P'\cdot(R'\|U')\leq Q'$, since $Q\leq P'$, monotonicity implies 
$P\cdot(R\|U)\cdot(R'\|U')\leq Q'$ and the last premise gives us $P\cdot[(R\cdot R')\|(U\cdot U')]\leq Q'$. The guaranteed part follows from monotonicity of $\cdot$.
\end{proof}

These rules support the construction of larger systems from the components using concurrent and/or sequential compositions. Together with these rely/guarantee rules, we will also make extensive use of Equation~(\ref{eq:pc-supdist}) because it provides the transport of probabilistic choices to the ``upper level\rq\rq{} of the specification automaton. That is, it allows us to write simple rely properties and postconditions of the form $(bad + good)\pc{p}good$ or even $bad\pc{p}good$ where $good$ and $bad$ are usually standard automata. With the testing interpretation of simulation, we conclude that the maximal probability for $bad$ to happen is bounded from above by $1-p$.

\begin{example}
Using our running example, we have the following rely/guarantee quintuples
$$\one\ \run\ \{M\}\ \run\ H\qquad\textrm{and}\qquad \coin \ H\ \{U'\}\ Q\ \run$$
where $Q$ is given by the diagram in Figure~\ref{fig:post}.
\begin{figure}
\begin{displaymath}
\xymatrix{&w_0\ar[d]_{\coin}& \\
&\arp[dl]|{0.04}\arp[dr]|{0.96}& \\
w_1\ar[d]_{\kick}&&w_2\ar[d]^{\kick,\tea}\\
w_3\ar[d]_{\kick}&& w_4\ar@(d,r)_{\tea}\\
w_5\ar@(d,l)^{\fail}&& 
}
\end{displaymath}
\caption{The postcondition for the system in the form $Q = \coin\cdot Q'$.}\label{fig:post}
\end{figure}
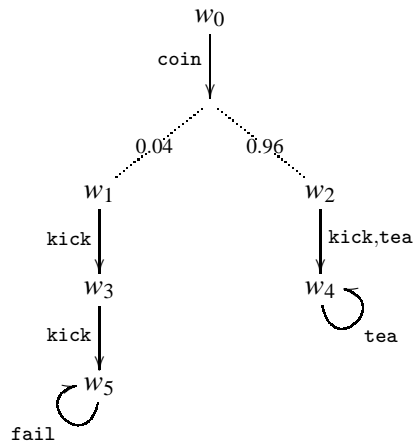

The first quintuple has been established algebraically in the previous section, that is, $M\leq H$. As for the second one, it is clear that $H\| U'\leq Q'$ which can be established by direct automata calculation or using the algebra as before.

Therefore, Rule~(\ref{rul:asym}) implies that the quintuple
$$\coin\ \run\ \{M\|U'\}\ Q\ H$$
is valid. That is, $V\| U = \coin\cdot(M\|U')\leq Q$~\footnote{Notice that $M\|U'$ does not enable any transition labelled by $\fail$ because that action has to be synchronised. But the established property says that fail can occur only with probability at most $0.04$, but not that it has to occur at all. If we wanted an explicit $\fail$, we can form self loops labelled by $\fail$ on each state of $V$.} which says that with probability at least $0.96$ the user needs to kick the machine at most once to get tea. (Note we have used the fact that $(\coin\cdot M)\|(\coin\cdot U\rq{})\equiv \coin\cdot(M\|U\rq{})$ which is a stronger version of the interchange law).
\end{example}

\section{Related Work}

This paper aims to develop an algebra that accounts for nondeterminism, probability and concurrency in a Kleene algebraic fashion~\cite{Con71,Hoa09,Koz94,Koz97,Mci04,Mo07}. To the best of our knowledge, there is no algebraic structure in the style of Kleene algebras that includes nondeterminism, probability and concurrency. The algebra we develop is a mixture of concurrent~\cite{Hoa09} and probabilistic~\cite{Mci04} Kleene algebras augmented with probabilistic choices to manipulate quantitative properties. The soundness of the algebra is established using probabilistic automata modulo simulation as in~\cite{Den07a}. That paper provides an extensive survey of the algebraic laws for such model in the style of a recursion-free process algebra, hence there is no sequential composition (which is mandatory for the encoding of interference) nor Kleene star (which provides a meaning for terminating loops). Another related work is the quantitative Kleene coalgebra of~\cite{Sil11}. That paper focuses on unifying various constructions of transition systems through the use of functor-coalgebras. It also provides a coalgebra composed of algebraic expressions though the main focus is on the generalisation of Kleene's correspondence between operational semantics and the expressions through the use of derivatives. Moreover, these expressions are generated from a signature that is different from the one we propose in this paper, most importantly, concurrency is not considered. Reconciling the two approaches seems very promising.

The algebraic approach to the rely/guarantee calculus of Section~\ref{sec:rely-guarantee} is a straightforward generalisation of~\cite{Hoa09} but now proved to be valid for probabilistic scenario as well. Our approach is conceptually related to~\cite{Kom12} where the rules are interpreted against probabilistic strong simulation. The precongruence of~\cite{Kom12} is not applicable in our setting because we need explicit internal actions to construct the algebraic operators and these internal actions cannot usually be removed in presence of probability and nondeterminism. Hence, strong simulation is inadequate for an algebraic approach that should be a generalisation of Kleene algebras. The disadvantage of using weak simulation is that efficient decidability is unknown in contrast to strong simulation~\cite{Kom12}. Though this is an interesting problem, 
our focus is on using the algebra for direct proof of the existence of a simulation rather than computing a simulation relation directly. Therefore, we can use theorem provers or proof assistants to handle the automation.

\section{Conclusion}

This paper presented a Kleene algebraic approach to systems exhibiting nondeterminism, probability and concurrency. A sound axiomatisation has been presented with respect to the set of probabilistic automata modulo probabilistic simulation. The simulation used is equivalent to the vector may testing preorder of~\cite{Den07a} which provides the interpretation of maximal probability of failure. The algebra was constructed as a combination of probabilistic and concurrent Kleene algebras. The nondeterminisms of both algebras coincide, probability is handled by the subdistributivity law of probabilistic Kleene algebra as well as the explicit probabilistic choices, and concurrency is mainly handled by a weakening of concurrent Kleene algebra. 

An important law of probability is summarised by the equation 
$$P\cdot (Q\pc{p} R) \leq (P\cdot Q)\pc{p} (P\cdot R).$$ In multiple cases, it allows us to ``move probabilities upward" and to write specification of the form  $(bad + good)\pc{p}good$. This expression says that the probability of executing $bad$ is bounded above by $p$ because of the testing interpretation of the simulation. In fact, as shown in the vending machine example, it is also possible to write properties such as ``the maximal probability of failing after $k$-steps is less than $p$\rq\rq{}.  

We note finally that the rely/guarantee calculus is valid for any model satisfying the axioms. This applies to true-concurrent models as well as fragments of the establised automata models as long as they have sequential composition instead of prefixing. However, the action-based interleaving model provides a limited application of the rely/guarantee calculus. Therefore, the construction of an alternative model of the algebra using true-concurrency semantics is part of our future investigation. 

%\nocite{*}
\bibliographystyle{eptcs}
\bibliography{./references-comp}

\newpage

\setcounter{theorem}{0}
\renewcommand{\thetheorem}{A.\arabic{theorem}}
    
\section*{Appendix}

\begin{proposition}\label{ap:pro:precongruence}
Simulation is a precongruence i.e. if $P\leq Q$ then $P+R\leq Q+R$, $P\cdot R\leq Q\cdot R$,  $P^*\leq Q^*$, $P\pc{p}R\leq P\pc{p}R$, $P\pr{A} R\leq Q\pr{A} R$ and the same holds for binary operators when the order of the arguments is reversed.
\end{proposition}

\begin{proof}
The construction of simulation in each of these cases is easy and they have been proven elsewhere~\cite{Den07a,Seg96} except for sequential composition and Kleene star which we present here.

Let $P,Q,R$ be probabilistic automata such that $S:P\to Q$ is a simulation. We show that the relation $$S'=\{(x,\nu)\ |\ (x\in P\wedge \nu\in \DD(Q)\wedge xS\nu)\vee(x\in R\wedge \nu= \delta_x)\}$$
is a simulation from $P\cdot R$ to $Q\cdot R$. Since the initial distributions of $P\cdot R$ and $Q\cdot R$ are respectively the initial distributions of $P$ and $Q$, Properties (a) and (c) of Definition~\ref{df:simulation} are immediate for this case. As for part (b), let $\trans{x}{a}{\mu}\in\lra_{P\cdot R}$ and $(x,\nu)\in S'$. There are three cases:
\begin{itemize}
\item[-] $\trans{x}{a}{\mu}\in\lra_P$, and then $(x,\nu)\in S$ and we are done because $S$ is a simulation from $P$ to $Q$ and $S\subseteq S'$.
\item[-] $\trans{x}{a}{\mu}\in\lra_R$, and then $\nu = \delta_x$ and we are done because $id_R$ is a simulation.
\item[-] $\trans{x}{\tau}{\psi_0}$ such that $x\in F_P$  and $\psi_0$ is the initial distribution of $R$. We are only interested in the case where $x$ is reachable from $\phi_0$. Since $P\leq Q$, there exist $\nu\in\DD(Q)$ and $\nu'\in\DD(F_Q)$  such that $x S\nu$, and $\Trans{\nu}{}{\nu'}$. Therefore, $\Trans{\nu}{}{\psi_0}$ is a valid weak transition in $Q\cdot R$ and we have $\psi_0\ov{id_R}\psi_0$.
\end{itemize}
The dual $R\cdot P\leq R\cdot Q$ also holds using an analogous relation.

Finally, we show that Kleene star is monotonic with respect to $\leq$. Let $S:P\to Q$ be a simulation and denote by $x_P,x_Q$ the respective initial states of $P^*$ and $Q^*$ and consider the relation 
$$S' = S\cup\{(x_P,\delta_{x_Q})\}.$$ 
Properties (a) and (b) are routine. Let $\trans{x}{a}{\mu}\in\lra_{P^*}$ and $xS'\nu$, there are two cases:
\begin{itemize}
\item[-] $xS\nu$, then $\trans{x}{a}{\mu}\in\lra_P$ or $x\in F_P$ and $\mu=\delta_{x_P}$. 
\begin{itemize}
\item If $\trans{x}{a}{\mu}\in\lra_P$ then there exists $\nu'\in\D(Q)$ such that $\Trans{\nu}{\ov a}{\nu'}$ and $\mu S\nu'$ because $S$ is simulation. Since $S\subseteq S'$, we deduce that $\mu S'\nu'$. 
\item If $x\in F_P$ and $\trans{x}{\tau}{\delta_{x_P}}$, we can assume that $x$ is a reachable final state of $P$ because we consider probabilistic automata to be simulation equivalent if and only if they reachable parts are. Since $S$ is simulation and $x S\nu$, we have $\Trans{\nu}{}{\nu'}$ for some $\nu'\in \DD(F_Q)$. Therefore, $\Trans{\nu}{}{\delta_{x_Q}}$ and $x_PS'\delta_{x_Q}$ by definition of $S'$. 
\end{itemize}
\item[-] $x_P S'\delta_{x_Q}$ and $\trans{x_P}{\tau}{\mu_0}$ where $\mu_0$ is the initial distribution of $P$. Since $S$ is a simulation, there exists $\nu_0'$ such that $\Trans{\nu_0}{}{\nu_0'}$ and $\mu_0\ov{S}\nu_0'$. Hence, $\Trans{\delta_{x_Q}}{}{\nu_0'}$ and $\mu_0S\rq{}\nu_0\rq{}$.  \qedhere
\end{itemize}
\end{proof}

\begin{proposition}\label{ap:pro:sequential}
The sequential composition satisfies Equations~(\ref{eq:right-skip}-\ref{eq:pc-supdist}).
\end{proposition}

\begin{proof}
Equations~(\ref{eq:right-skip}) and~(\ref{eq:left-skip}) are clear and~(\ref{eq:left-zero}) follows form the fact that $P\equiv Q$ iff their reachable parts are simulation equivalent.

Associativity~(\ref{eq:.-assoc}) is evident because the left and right hand side automata are exactly the same.

For distributivity~(\ref{eq:+-right-dist}), let us write the left hand side term as 
$P\cdot R+Q\cdot R_c$ where $R_c$ is a copy of $R$ whose states are renamed to $x_c$ for every state $x$ of $R$. We construct a relation $S\subseteq(P\cup Q\cup\{z\}\cup R\cup R_c)\times \DD(P\cup Q\cup\{z\}\cup R)$ such that $S = \{(x,\delta_x), (x_c,\delta_x)\ |\ x\in R\wedge x_c\textrm{ is the copy of } x\}\cup id_{P+Q}$. It is easy to shown that $S$ is a simulation and so is its inverse.

For subdistributivity~(\ref{eq:+-left-subdist}), we consider the relation 
$$S = \{(x,\delta_x), (x_c,\delta_x)\ |\ x\in P \wedge x_c\textrm{ is the copy of } x\}\cup\{(z,\mu_0)\}\cup id_Q\cup id_R$$
where $z$ is the initial state of $P\cdot Q+P\cdot R$ and $\mu_0$ is the initial distribution of $P$. It is again straightforward to prove that $S$ is indeed a simulation. 

Equation~(\ref{eq:pc-dist}) is proved using the exact same simulation constructed in the case of Equation~(\ref{eq:+-right-dist}).

For the last equation, let 
$$S = \{(x,\delta_x\pc{p}\delta_{x_c})\ |\ x\in P\wedge x_c\textrm{ is the copy of } x\}\cup id_Q\cup id_R$$
This simulation essentially says that we carry down the probabilistic choice $\pc{p}$ on the left hand side until it needs to be resolved.
\begin{itemize}
\item[-] By construction of the simulation, we have $\mu_0\ov{S}(\mu_0\pc{p}\mu_{0c})$ where $\mu_0$ and $\mu_{0c}$ are the respective initial distributions of $P$ and $P_c$.
\item[-] Let $\trans{x}{a}{\mu}\in\lra_{P\cdot(Q\pc{p} R)}$ and $xS\nu$, there are three cases
\begin{itemize}
\item $\trans{x}{a}{\mu}\in\lra_{P}$, therefore $\nu = \delta_x\pc{p}\delta_{x_c}$ and $\trans{\nu}{\ov{a}}{\mu\pc{p}\mu_c}$ where $\mu_c$ is the copy of $\mu$.
\item $\trans{x}{a}{\mu}\in\lra_{Q}\cup\lra_R$, then we are done because $id_Q\cup id_R\subseteq S$.
\item $\trans{x}{\tau}{\mu_{0Q}\pc{p}\mu_{0R}}$ and $x\in F_P$, then $\trans{x}{\tau}{\mu_{0Q}}$ and $\trans{x}{\tau}{\mu_{0R}}$ are valid transitions of $P\cdot Q$ and $P_c\cdot R$. But $\nu = \delta_x\pc{p}\delta_{x_c}$ because $x\in P$, therefore $\trans{\nu}{\ov{\tau}}{\mu_{0Q}\pc{p} \mu_{0R}}$.
\end{itemize}
\item[-] Let $xS\nu$ and $x$ is a final state. By definition of $\pc{p}$, $x\in F_Q\cup F_R$ and hence $\nu = \delta_x\in\DD(F_Q)\cup\DD(F_R)$.
\end{itemize}
\end{proof}

Before we prove that the Kleene star satisfies the usual unfold and left induction law of probabilistic Kleene algebra, let us introduce the notion of unfolding which will simplify the proof of the induction law considerably. It is essentially a cleaner version of our construction in~\cite{Rab11}. We denote $\unfold(P)$ the unfold of any automaton $P$~\cite{Seg03}, that is, the automaton
$$(\Path(P),\lra,\mu_0,F)$$
where 
$$\lra = \{\trans{\alpha}{a}{\mu}\ |\ \alpha\in\Path(P)\wedge\exists\mu\rq{}\in\DD(P):\trans{\last(\alpha)}{a}{\mu\rq{}}\in\lra_P\wedge\mu(\alpha a x) = \mu\rq{}(x)\}$$
and 
$$F = \{\alpha\in\Path(P)\ |\ \last(\alpha)\in F_P\}.$$
and $\mu_0$ is the initial distribution of $P$.
This construction provides us with an automaton whose states are finite paths in $P$ and there is a transition between two paths $\alpha,\alpha'$ iff $\alpha' = \alpha ax$ where $a\in\Sigma_\tau$ and $x\in P$. Such a transition is labelled by $a$. It is now easy to show that the relation $\{(\alpha,\delta_{\last(\alpha)})\ |\ \alpha\in\Path(P)\}$ is a simulation from $\unfold(P)$ to $P$ and the inverse is also a simulation from $P$ to $\unfold(P)$~\cite{Seg03}.

\begin{proposition}\label{ap:pro:kleene-star}
Kleene star satisfies Equation~(\ref{eq:*-unfold}) and the induction law~(\ref{eq:*-induction}).
\end{proposition}
\begin{proof}
Let $u$ be the initial state of $\one + P\cdot P^*$ and $v$ be the initial state of $P^*$. Since we add only one state and some transition in the construction of $P^*$, we denote $x_*\in P^*$ the state corresponding to $x\in P$. To prove Equation~(\ref{eq:*-unfold}), we consider the relation 
$$S = \{(x_*,\delta_{x_*}),(x_*,\delta_x)\ |\ x\in P\}\cup \{(v,\delta_v),(v,\delta_u)\}$$
from $P^*$ to $\one + P\cdot P^*$. We now prove that $S$ is a simulation.
\begin{itemize}
\item[-] For the initial distribution, we have $v S\delta_u$.
\item[-] Let $\trans{y}{a}{\mu}$ be a valid transition in $P^*$ and $yS\nu$. There are two cases
\begin{itemize}
\item $y = v$ and the transition is $\trans{v}{\tau}{\mu_0}$ where $\mu_0$ is the initial distribution of $P$. If $\nu = \delta_v$ then we are done because $\{(v,\delta_v)\}\cup\{(x_*,\delta_{x_*})\ |\ x\in P\}=id_{P^*}$. Else, $\nu = \delta_u$ and $\trans{u}{\tau}{\mu_0}$ is a valid transition in $\one + P\cdot P^*$.
\item $y = x_*$ for some $x\in P$ and:
\begin{itemize}
\item $\trans{x_*}{a}{\mu_*}$ is the copy of a transition of $P$. Therefore, if $\nu = \delta_{x_*}$ then the same transition belongs to $P\cdot P^*$. If $\nu = \delta_x$ then $\trans{x}{a}{\mu}$ is a transition of $P$ and $\mu_*\ov{S}\mu$.
\item or, $\trans{x_*}{\tau}{\delta_v}$ and in this case, if $\nu = \delta_{x^*}$ then that transition belongs to $P\cdot P^*$ again, else $\nu = \delta_x$ and $x\in F_P$. Therefore, $\trans{\delta_x}{\ov{\tau}}{\delta_v}$ is a lifted transition in $P\cdot P^*$.
\end{itemize}
\end{itemize}
\item[-] The conservation of final state is obvious because $F_{P^*} = \{v\}$ and $\trans{u}{\tau}{\delta_z}$ where $z$ is the final state of $\one$ in $\one+P\cdot P^*$.
\end{itemize}

With the similar reasoning, it holds that the inverse of $S$ is a simulation from $\one+P\cdot P^*$ to $P^*$.

We now prove the induction law~(\ref{eq:*-induction}). We can assume that $P$ is loop-free by unfolding it and therefore $1 + P\cdot\unfold(P^*)$ is again loop-free and simulation equivalent to $P^*$. Let $\F(X) = \one + P\cdot X$. Since $P\cdot\zero\leq P$, we show easily by induction that $\unfold(\F^n(\zero))\trianglelefteq \unfold(\F^{n+1}(\zero))$ where $\trianglelefteq$ is the inclusion of automata i.e. $X\trianglelefteq Y$ if the state space of $X$ is a subset of the state space of $Y$, transitions of $X$ are transitions of $Y$ and $F_X\subseteq F_Y$. We can then construct a limit automta $\lim_n\F^n(\zero) = \F^*(\zero)$ obtained as the countable union of component by component (the set of states is the union of the sets of states, the set of transitions is the union of sets of transitions,\dots). Since $P$ has no cycle, it follows that $\F^*(\zero) = \unfold(P^*)$.

Now assume that $P\cdot Q\leq Q$, then $(\one + P\cdot\zero)\cdot Q\leq (\one + P)\cdot Q\leq Q$ and by induction, $\F^n(\zero)\cdot Q\leq Q$ for every $n\in\mathbb{N}$. Moreover, since $\unfold(\F^n(\zero))\trianglelefteq\unfold(\F^{n+1}(\zero))$, we have $\unfold(\F^n(\zero))\cdot Q\trianglelefteq\unfold(\F^{n+1}(\zero))\cdot Q$ and since $F_{\unfold(\F^{n}(\zero))}\subseteq F_{\unfold(\F^{n+1}(\zero))}$ (inclusion of final states), $\lim_n(\unfold(\F^n(\zero))\cdot Q) = \F^*(\zero)\cdot Q$ (the two automaton are equal by construction). Hence $\F^*(\zero)\cdot Q\leq Q$.
\end{proof}
\begin{proposition}\label{ap:pro:parallel}
The parallel composition satisfies Equations~(\ref{eq:|-comm}-\ref{eq:|-distributes-pc}).
\end{proposition}

\begin{proof}
Equations~(\ref{eq:|-comm}), (\ref{eq:|-subdistributes-+}) and $(\ref{eq:|-distributes-pc})$ are proven in~\cite{Den07a}. 

For the associativity, remind that when the frame is fixed then there is a standard simulation between $P\|(Q\|R)$ and $(P\|Q)\|R$ by associating each tuple $(x,(y,z))$ to $((x,y),z)$. That simulation is lifted to $(P\times (Q\times R))\times\DD((P\times Q)\times R)$ using point distributions and dually.

As for the interchange law~(\ref{eq:exchange-law}), we consider the injection 
$$S = \{((x,y),\delta_{(x,y)})\ |\ (x,y)\in (P\times Q)\cup (P'\times Q')\}$$
from $U = (P\| Q)\cdot (P'\| Q')$ to $V = P\cdot P'\|Q\cdot Q'$.
\begin{itemize}
\item[-] Using the definition of $\|$ and $\cdot$, we deduce that the initial distributions of $U$ and $V$ are the same.
\item[-] Let $(x,y)S\delta_{(x,y)}$ and $\trans{(x,y)}{a}{\mu}\in\lra_U$. There are three cases:
\begin{itemize}
\item $(x,y)\in P\times Q$ and $\mu = \mu_P\times\mu_Q\in\DD(P\times Q)$. In all three cases in the definition of $\|$, we have $\trans{(x,y)}{a}{\mu_P\times\mu_Q}\in\lra_V$.
\item $(x,y)\in P'\times Q'$ and $\mu = \mu_{P'}\times\mu_{Q'}\in\DD(P'\times Q')$. This is the same as the previous case.
\item $(x,y)\in F_P\times F_Q$ and the transition is $\trans{(x,y)}{\tau}{\mu_{0P'}\times\mu_{0Q'}}$ where $\mu_{0P'},\mu_{0Q'}$ are the respective initial distributions of $P',Q'$. Since $x\in F_P$, $\trans{x}{\tau}{\mu_{0P'}}\in\lra_{PP'}$ and similarly for $y\in F_Q$. Therefore, $\trans{(x,y)}{\tau}{\trans{\mu_{0P'}\times\delta_y}{\ov{\tau}}{\mu_{0P'}\times\mu_{0Q'}}}$ i.e. $\Trans{(x,y)}{}{\mu_{0P'}\times\mu_{0Q'}}$ is a weak lifted transition in $V$.
\end{itemize}
\item[-] Finally, $F_U = F_V$.\qedhere
\end{itemize}
\end{proof}

\end{document}